\documentclass[12pt,reqno,a4paper]{amsart}

\date{March 6, 2023}

\usepackage{amsmath,amsfonts,amsthm,amssymb,amsxtra,enumerate,mathtools,mathabx}
\usepackage[normalem]{ulem}
\usepackage[utf8]{inputenc}
\usepackage[T1]{fontenc}
\usepackage{lmodern}
\usepackage{bbm}
\usepackage{mathrsfs}
\usepackage{enumerate}
\usepackage{comment}
\usepackage{tikz}
\usepackage{pgfplots}
\usepackage{bbm} 

\newtheorem{theorem}{Theorem}

\newtheorem{lemma}[theorem]{Lemma}
\newtheorem{corollary}[theorem]{Corollary}
\theoremstyle{definition}

\theoremstyle{remark}

\newtheorem{remark}[theorem]{Remark}

\newcommand{\1}{\mathbbm{1}}

\newcommand{\dd}{\, \mathrm{d}}

\renewcommand{\epsilon}{\varepsilon}

\newcommand{\ii}{\mathrm{i}}

\newcommand{\N}{\mathbb{N}}

\renewcommand{\phi}{\varphi}
\newcommand{\R}{\mathbb{R}}

\newcommand{\Sph}{\mathbb{S}}

\newcommand{\Z}{\mathbb{Z}}

\DeclareMathOperator{\Tr}{Tr}

\DeclareMathOperator{\DeltaAnti}{\Delta_{{\bf as}}}

\begin{document}
	
	\title[Weighted CLR bounds
	]
	{Weighted CLR type bounds in two dimensions}
 
    \author{Rupert L. Frank} 
	\address[Rupert L. Frank]{Mathematisches Institut, Ludwig-Maximilans Universit\"at M\"unchen, Theresienstr. 39, 80333 M\"unchen, Germany, and Munich Center for Quantum Science and Technology, Schellingstr. 4, 80799 M\"unchen, Germany, and Mathematics 253-37, Caltech, Pasadena, CA 91125, USA}
	\email{r.frank@lmu.de}
    
    \author{Ari Laptev} 
	\address[Ari Laptev]{Department of Mathematics, Imperial College London, London SW7 2AZ, UK, and Sirius Mathematics Center, Sirius University of Science and Technology, 1 Olympic Ave, 354340, Sochi, Russia}
	\email{a.laptev@imperial.ac.uk}

    \author{Larry Read} 
	\address[Larry Read]{Department of Mathematics, Imperial College London, London SW7 2AZ, UK}
	\email{l.read19@imperial.ac.uk}
	
    \begin{abstract}
        We derive weighted versions of the Cwikel--Lieb--Rozenblum inequality for the Schr\"odinger operator in two dimensions with a nontrivial Aharonov--Bohm magnetic field. Our bounds capture the optimal dependence on the flux and we identify a class of long-range potentials that saturate our bounds in the strong coupling limit. We also extend our analysis to the two-dimensional Schr\"odinger operator acting on antisymmetric functions and obtain similar results.
    \end{abstract}
    
    \maketitle
    
\section{Introduction and main results}
The celebrated Cwikel--Lieb--Rozenblum (CLR) inequality states that the number $N(-\Delta-V)$ of negative eigenvalues, including multiplicity, of a Schr\"odinger operator $-\Delta-V$ in $L^2(\R^d)$ in dimension $d\geq 3$ is bounded by 
\begin{equation}\label{eqn:CLR}
    N(-\Delta-V)\lesssim_d \int_{\R^d}V(x)_{+}^{d/2}\dd x
\end{equation}
where the implied constant is independent of $V$. Here and throughout we take $a_{\pm}\coloneqq\max(0,\pm a)$ and use a subscript on $\lesssim$ to specify the variables on which the implied constant depends. The inequality is due to M.~Cwikel \cite{Cwikel1977WeakOperators}, E.~Lieb \cite{Lieb1976BoundsOperators} and G.~Rozenblum \cite{Rozenbljum1972DistributionOperators}. For further proofs and background we direct the reader to \cite{Frank2022SchrodingerInequalities}. The bound is saturated in the strong coupling limit, that is where $V$ is replaced with $\lambda V$ and $\lambda\rightarrow\infty$, since by Weyl's asymptotics,
\begin{equation}\label{eqn:Weyl}
    \lim_{\lambda\rightarrow\infty}\lambda^{-d/2}N(-\Delta-\lambda V)=\frac{\omega_d}{(2\pi)^d}\int_{\R^d}V(x)_{+}^{d/2}\dd x,
\end{equation}
where $\omega_d$ is the volume of the unit ball in $\R^d$. One of the uses of \eqref{eqn:CLR} is to extend this asymptotic behavior, which is originally established for instance for continuous $V$ of compact support, to all $V$ with $V_+ \in L^{d/2}(\R^d)$. Concerning the repulsive part one only needs to assume $V_-\in L^1_{\rm loc}(\R^d)$ \cite{Fr}.

Building on earlier work for radial potentials by V.~Glaser, H.~Grosse and A.~Martin \cite{Glaser1978BoundsOperator}, the CLR inequality was generalised by Y.~Egorov and V.~Kondratiev in \cite{Egorov1989OnOperator} to include the weighted bounds
\begin{equation} \label{eqn:weightedCLR}
    N(-\Delta-V)\lesssim_{d,\alpha}\int_{\R^d}V(x)_{+}^{(d+\alpha)/2} |x|^\alpha \dd x,
\end{equation}
which hold in dimensions $d\geq 3$ for any $\alpha> 0$. In \cite{Birman1988InterpolationOperator}, M.~Birman and M.~Solomyak showed that the strong $L^p$ norm appearing on the right in \eqref{eqn:weightedCLR} can be replaced by a weak norm, namely
\begin{equation} \label{eqn:weightedCLRweak}
    N(-\Delta-V)\lesssim_{d,\alpha} \sup_{t>0} t^{(d+\alpha)/2}\int_{|x|V(x)_{+}>t}\frac{\dd x}{|x|^d},
\end{equation}
which is valid, again, in dimensions $d\geq 3$ with $\alpha>0$. Note that the bounds \eqref{eqn:weightedCLR} and \eqref{eqn:weightedCLRweak} are homogeneous with respect to $V$ of degree $(d+\alpha)/2>d/2$, in contrast to the homogeneity $d/2$ of \eqref{eqn:CLR}. The latter homogeneity is consistent with \eqref{eqn:Weyl}. Nevertheless, as shown by M.~Birman and M.~Solomyak \cite{Birman1990NegativeDiscussion}, the asymptotic order of growth $(d+\alpha)/2$ in \eqref{eqn:weightedCLRweak} can be saturated in the strong coupling limit for a class of potentials with particular long range behaviour. Namely, if $V_{+}\in L^{d/2}_{\rm loc}(\R^d)$ satisfies
\begin{equation}\label{eqn:saturatingpotential}
    V(x)= |x|^{-2}|\ln|x||^{-1/p} \left( 1+ o(1) \right) \text{ as }|x|\rightarrow\infty
\end{equation}
for some $p>d/2$, then one can show that 
\begin{equation*}
    \lim_{\lambda\rightarrow\infty}\lambda^{-p}N(-\Delta-\lambda V) \text{ exists and is finite,}
\end{equation*}
while for $\alpha>0$ with $p=(d+\alpha)/2$, 
\begin{equation*}
    \lim_{\lambda\rightarrow\infty}\lambda^{-p}\sup_{t>0}t^{(d+\alpha)/2}\int_{\lambda |x|^2V(x)_+>t}\frac{\dd x}{|x|^d}=\sup_{t>0} t^{p}\int_{ |x|^2V(x)_+>t}\frac{\dd x}{|x|^d} \in (0,\infty).
\end{equation*}

All the results discussed so far are restricted to the case of dimensions $d\geq 3$ and most of their direct analogues in dimensions $d=2$ fail. For instance, none of the direct analogues of \eqref{eqn:CLR}, \eqref{eqn:weightedCLR} and \eqref{eqn:weightedCLRweak} hold. Moreover, there are examples of $V\in L^1(\R^2)$ with $V\geq 0$ for which either the limit on the left side of \eqref{eqn:Weyl} is infinite or it is finite but different from the right side, see \cite{Birman1996TheOperator}. Recently, there has been a lot of activity in proving bounds on $N(-\Delta-V)$ in $d=2$ and in giving necessary and sufficient conditions for either the bound $\lim_{\lambda\rightarrow \infty}\lambda^{-1}N(-\Delta-\lambda V)<\infty$ or the validity of \eqref{eqn:Weyl}. A sample of references for this development is \cite{KMW,MoVa,Shargorodsky2014OnOperators,GrNa,LaSo1,LaSo2}. An earlier fundamental paper is due to M.~Solomyak \cite{So}; see also \cite{FrLa}.

In this paper we are concerned with bounds on the number of negative eigenvalues of two-dimensional Schr\"odinger operators in the presence of an Aharonov--Bohm magnetic field. We will see that when this field is nontrivial, one obtains inequalities that are analogous to those discussed above for Schr\"odinger operators in dimensions $d\geq3$ and see that the difficulties of the two-dimensional case mostly disappear. We will also consider the case of the non-magnetic Schr\"odinger operator restricted to antisymmetric functions and see that this case is similar to that of an Aharonov--Bohm magnetic field.

Our results support the heuristics that the different behaviour in dimensions $d\geq 3$ and in $d=2$ comes from a spectral instability of the two-dimensional Laplacian near energy zero and that this instability can be removed by additional repulsion, either in the form of a magnetic field or the presence of symmetries. For other instances of this principle see \cite{LaNe,Laptev2022CalogeroDimensions}.

\medskip

To be more specific, let
$$
{\bf A}(x)= |x|^{-2}(x_2,-x_1)
\qquad\text{for all}\ x=(x_1,x_2)\in\R^2
$$ 
and for $\Phi\in\R$ let
$$
D_\Phi = -i\nabla +\Phi {\bf A} \,.
$$
We consider the magnetic Schr\"odinger operators
\begin{equation*}
	D_\Phi^2-V \text{ in }L^2(\R^2) \,.
\end{equation*}
As discussed in the next section, under suitable conditions on $V$ this operator can be realized as a self-adjoint operator via the closure of the corresponding quadratic form on $C_0^\infty(\R^2\backslash\{0\})$. When $\Phi\in\Z$, the magnetic potential can be gauged away and the operator is unitarily equivalent to $-\Delta+V$. Therefore, in the following we will concentrate on the case $\Phi\in\R\setminus\Z$.

An analogue of the CLR inequality \eqref{eqn:CLR} was shown by A.~Balinsky, W.~Evans and R.~Lewis \cite{Balinsky2001OnField}, namely,
\begin{equation}\label{eqn:ABCLRBalinsky}
    N(D_\Phi^2-V)\lesssim_{\Phi} \int_0^\infty \sup_{\omega\in\Sph}V(r\omega)_{+}r\dd r.
\end{equation}
More recently it was deduced in \cite{Laptev2022CalogeroDimensions} that when $V_+$ is radially non-increasing one can replace the supremum over angles in the right side of \eqref{eqn:ABCLRBalinsky} with an integral, that is,
\begin{equation}\label{eqn:ABCLRLaptev}
    N(D_\Phi^2 -V)\lesssim_{\Phi}\int_{\R^2} V(x)_+\dd x.
\end{equation}
However, it is known \cite{Balinsky2001OnField} that this replacement cannot be made for general $V\in L^1(\R^2)$.

Our main result is the following magnetic version of \eqref{eqn:weightedCLR}.

\begin{theorem}\label{thm:ABweightedCLR}
    Let $\Phi\in\R\setminus\Z$ and $\alpha>0$. Then there is a constant $C_{\Phi,\alpha}<\infty$ such that
    \begin{equation}\label{eqn:ABweightedCLR}
        N(D_\Phi^2-V)\leq C_{\Phi,\alpha} \int_{\R^2}V(x)_+^{1+\alpha/2}|x|^{\alpha}\dd x
    \end{equation}
	for all $V\in L^1_{\rm loc}(\R^2)$ for which the right side is finite. Moreover, the optimal constant in this inequality satisfies
	\begin{equation}
		\label{eq:abequiv}
			C_{\Phi,\alpha} \sim_\alpha d(\Phi)^{-1-\alpha}
	\end{equation}
	with $d(\Phi):=\min_{k\in\Z} |\Phi-k|$.
\end{theorem} 

In fact, our proof yields the explicit upper bound
\begin{equation}
	\label{eq:abexplicit}
	C_{\Phi,\alpha} \leq \frac{\Gamma((1+\alpha)/2)}{4\pi^{3/2}\Gamma(1+\alpha/2)}\sum_{n\in\Z} |n-\Phi|^{-1-\alpha} \,.
\end{equation}
From this bound we immediately obtain the upper bound $C_{\Phi,\alpha} \lesssim_\alpha d(\Phi)^{-1-\alpha}$ in \eqref{eq:abequiv}. In the proof of Theorem \ref{thm:ABweightedCLR} we will show that this bound is sharp, thereby obtaining the precise divergence of the constant as the flux $\Phi$ approaches an integer value.

We complement Theorem \ref{thm:ABweightedCLR} with a variant of this bound with a weak norm.

\begin{corollary}\label{cor:ABweightedCLRweak}
    Let $\Phi\in\R\setminus\Z$ and $\alpha>0$. Then there is a constant $C_{\Phi,\alpha}^\prime<\infty$ such that
    \begin{equation}\label{eqn:ABweightedCLRweak}
        N(D_\Phi^2 -V)\leq C_{\Phi,\alpha}^\prime \sup_{t>0}t^{1+\alpha/2}\int_{|x|^2 V(x)_+>t}\frac{\dd x}{|x|^2}
    \end{equation}
for all $V\in L^1_{\rm loc}(\R^2)$ for which the right side is finite. Moreover, the constant can be chosen to satisfy
\begin{equation}
	\label{eq:abequivweak}
	C_{\Phi,\alpha}^\prime \sim_\alpha d(\Phi)^{-1-\alpha} \,. 
\end{equation}
\end{corollary}

Since
$$
\sup_{t>0}t^{1+\alpha/2}\int_{|x|^2 V(x)_+>t}\frac{\dd x}{|x|^2} \leq
\int_{\R^2}V(x)_+^{1+\alpha/2}|x|^{\alpha}\dd x \,,
$$
the bound \eqref{eqn:ABweightedCLR} follows from \eqref{eqn:ABweightedCLRweak} and for the sharp constants we find
\begin{equation}
	\label{eq:strongweak}
	C_{\Phi,\alpha} \leq C_{\Phi,\alpha}' \,.
\end{equation}
We will argue differently, however, and deduce Corollary \ref{cor:ABweightedCLRweak} from Theorem \ref{thm:ABweightedCLR}. To do this, we use an interpolation argument in the spirit of one of M.~Birman and M.~Solomyak~\cite{Birman1988InterpolationOperator}.

In further likeness to the situation for $-\Delta-V$ in dimensions $d\geq 3$, we derive examples of potentials with the same long-range behaviour \eqref{eqn:saturatingpotential} which saturate the weak inequality \eqref{eqn:ABweightedCLRweak} in the strong coupling limit. We refer to Section \ref{sec:asymptotics} for the details. There we will show, in particular,
\begin{equation}
	\label{eq:abexplicitweak}
	C_{\Phi,\alpha}' \geq \frac{\Gamma((1+\alpha)/2)}{4\pi^{3/2}\Gamma(1+\alpha/2)}\sum_{n\in\Z} |n-\Phi|^{-1-\alpha} \,,
\end{equation}
which should be compared with \eqref{eq:abexplicit}. Of course, these two bounds are consistent with \eqref{eq:strongweak}.

\medskip

Next, we describe our results for two-dimensional Schr\"odinger operators acting on antisymmetric functions. For functions $V$ on $\R^2$ that are symmetric in the sense that $V(x_1,x_2)=V(x_2,x_1)$ for almost every $x\in\R^2$ we can consider the operator $-\Delta - V$ in $L^2(\mathbb{R}^2)$ restricted to antisymmetric functions, that is, in the Hilbert space $$
L^2_{\mathbf{as}}(\mathbb{R}^2) = \{ u \in L^2(\mathbb{R}^2) : u(x_1, x_2) = - u(x_2, x_1) 
\ \text{for almost every}\ x\in\R^2 \} \,.
$$
We denote the resulting operator by $-\DeltaAnti - V$. Under the assumption that $V$ is radially non-increasing, a corresponding version of the CLR inequality for this operator was found in \cite{Laptev2022CalogeroDimensions}, namely
\begin{equation*}
    N(-\DeltaAnti-V)\lesssim \int_{\R^2}V(x)_{+}\dd x.
\end{equation*}
However, this inequality does not hold for general $V$, as noted in \cite[Remark 1]{Laptev2022CalogeroDimensions}.

Our second pair of main results are strong and weak weighted CLR bounds for $-\DeltaAnti - V$, analogous to the bounds we derived for the magnetic operator.

\begin{theorem}\label{thm:antiweightedCLR}
    Let $\alpha>0$, then there is a constant $C_{\alpha}<\infty$ such that 
    \begin{align}\label{eqn:AntiweightedCLR}
        N(-\DeltaAnti-V)\leq C_{\alpha}\int_{\R^2} V(x)_{+}^{1+\alpha/2} |x|^{\alpha}\dd x
    \end{align}
	for all symmetric $V\in L^1_{\rm loc}(\R^2)$ for which the right side is finite.
\end{theorem}

In fact, our proof yields the explicit upper bound
\begin{equation}
	\label{eq:antiexplicit}
	C_{\alpha} \leq \frac{\Gamma((1+\alpha)/2)}{2\pi^{3/2}\Gamma(1+\alpha/2)}\,\zeta(1+\alpha) \,,
\end{equation}
where $\zeta$ is the Riemann zeta function.

\begin{corollary}\label{cor:AntiweightedCLRweak}
    Let $\alpha>0$, then there is a constant $C_{\alpha}^\prime<\infty$ such that
    \begin{equation}\label{eqn:AntiweightedCLRweak}
        N(-\DeltaAnti-V)\leq C_{\alpha}^\prime \sup_{t>0}t^{1+\alpha/2}\int_{|x|^2 V(x)_+>t}\frac{\dd x}{|x|^2}
    \end{equation}
	for all symmetric $V\in L^1_{\rm loc}(\R^2)$ for which the right side is finite.
\end{corollary}

Again, for long-range potentials of the form \eqref{eqn:saturatingpotential} the bound in the corollary can be saturated in the strong coupling limit and one obtains the lower bound
$$
C_{\alpha}^\prime \geq \frac{\Gamma((1+\alpha)/2)}{2\pi^{3/2}\Gamma(1+\alpha/2)}\,\zeta(1+\alpha) \,.
$$

Our plan for the paper is as follows: In Section \ref{sec:proofweightedCLR} we present the proof of Theorems~\ref{thm:ABweightedCLR} and~\ref{thm:antiweightedCLR}. In Section \ref{sec:interpolation} we derive the weak forms of the inequalities above. Finally, in Section \ref{sec:asymptotics} we will show that these bounds are saturated in the strong coupling limit by potentials with long range behaviour \eqref{eqn:saturatingpotential}. 

\section{Proof of theorems \ref{thm:ABweightedCLR} and \ref{thm:antiweightedCLR}}\label{sec:proofweightedCLR}

\subsection{The Aharonov--Bohm operator}

We begin by showing that the operators $D_\Phi^2-V$ are well-defined in quadratic form sense when $\Phi\in\R\setminus\Z$ and $V$ is such that the right side in either Theorem \ref{thm:ABweightedCLR} or Corollary \ref{cor:ABweightedCLRweak} is finite. The main ingredient in this argument is the magnetic Hardy--Sobolev inequality
\begin{equation}\label{eqn:ABHardySobolev}
	\int_{\R^2}| D_\Phi u|^2\dd x \geq S_{\Phi,q} \left(\int_{\R^2}\frac{|u|^q}{|x|^2}\dd x\right)^{2/q}
	\qquad\text{for all}\ u\in C^\infty_0(\R^2\setminus\{0\}) \,,
\end{equation}
with $S_{\Phi,q}>0$ provided that $q\in[2,\infty)$. A proof of this inequality can be found in \cite[Section 3.1, Step 1]{Bonheure2020SymmetryFields} based on the diamagnetic inequality and a special case of the Caffarelli--Kohn--Nirenberg inequality for scalar functions. Alternatively, one can deduce this inequality using the method of \cite{Egorov1989OnOperator}. In the special case $q=2$ inequality \eqref{eqn:ABHardySobolev} with sharp constant is due to \cite{Laptev1999HardyForms} and reads
\begin{equation}\label{eqn:ABHardy}
	\int_{\R^2}| D_\Phi u|^2\dd x \geq d(\Phi)^{2} \int_{\R^2} \frac{|u|^2}{|x|^2}\dd x  \,.
\end{equation}
Some results about the sharp constant in \eqref{eqn:ABHardySobolev} for $q>2$ can be found in \cite{Bonheure2020SymmetryFields}.

Let us show how to use \eqref{eqn:ABHardySobolev} to define the operator $D_\Phi^2-V$. We combine \eqref{eqn:ABHardySobolev} with H\"older's inequality to obtain for $u\in C^\infty_c(\R^2\setminus\{0\})$
\begin{align}\label{eq:hardysobholder}
	\int_{\R^2} V|u|^2\,\dd x & \leq \left( \int_{\R^2} V_+^{1+\alpha/2} |x|^\alpha \dd x \right)^{1/(1+\alpha/2)} \left( \int_{\R^2} \frac{|u|^q}{|x|^2} \dd x \right)^{2/q} \notag \\
	& \leq S_{\Phi,q}^{-1} \left( \int_{\R^2} V_+^{1+\alpha/2} |x|^\alpha \dd x \right)^{1/(1+\alpha/2)} \int_{\R^2}| D_\Phi u|^2\dd x \,,
\end{align}
where $q$ and $\alpha$ are related by $1/(1+\alpha/2) + 2/q = 1$. The assumption $q<\infty$ is equivalent to $\alpha>0$.

Now given $V\in L^1_{\rm loc}(\R^2)$ such that the integral in Theorem \ref{thm:ABweightedCLR} is finite and given $\epsilon>0$, we decompose $V=V_1+V_2$ with $V_2\in L^\infty(\R^2)$ and $V_1\geq 0$ satisfying
$$
\int_{\R^2} V_1^{1+\alpha/2} |x|^\alpha \dd x \leq \epsilon.
$$
Applying \eqref{eq:hardysobholder} with $V_1$ we find that $V$ is form-bounded with respect to $D_\Phi^2$ relative form bound zero. This allows us to define $D_\Phi^2-V$ as a selfadjoint, lower semibounded operator in $L^2(\R^2)$ with form core $C_c^\infty(\R^2\setminus\{0\})$.

Meanwhile, let $V\in L^1_{\rm loc}(\R^2)$ be given such that the integral in Corollary \ref{cor:ABweightedCLRweak} is finite and let $\epsilon>0$. We choose $\tilde q\in(q,\infty)$ and define $\tilde\alpha>0$ by $1/(1+\tilde\alpha/2) + 2/\tilde q=1$. We can decompose $V=V_1+V_2$ with $\| |x|^2 (V_2)_+ \|_{L^\infty(\R^2)}\leq \epsilon$ and $V_1\geq 0$ satisfying
$$
\int_{\R^2} V_1^{1+\tilde\alpha/2} |x|^{\tilde\alpha} \dd x <\infty \,.
$$
(Indeed, we can simply take $V_1 = |x|^{-2} (|x|^2 V - \epsilon)_+ -V_-$.) Proceeding as before to control the $V_1$ piece and using \eqref{eqn:ABHardy} to control the $V_2$ piece, we find again that $V$ is form-bounded with respect to $D_\Phi^2$ with relative bound zero and, consequently, that $D_\Phi^2-V$ is well-defined.

Next, we recall that the operators $D_\Phi^2-V$ and $D_{\Phi-k}^2-V$ are unitarily equivalent for $k\in\Z$ and that the operators $D_\Phi^2-V$ and $D_{-\Phi}^2-V$ are antiunitarily equivalent; see, e.g., \cite[Subsection 2.1]{Bonheure2020SymmetryFields}. Thus, in what follows we can restrict ourselves to the case $\Phi\in(0,1/2]$.

We are now ready to present the proof of the weighted CLR bound for $D_\Phi^2-V$.

\begin{proof}[Proof of Theorem \ref{thm:ABweightedCLR}]
    Fix $\alpha>0$ and let $V_+|x|^2\in L^{1+\alpha/2}(\R^2;\dd x/|x|^2)$. As explained above, we may assume $\Phi\in(0,1/2]$. Moreover, by the variational principle, we may assume $V\geq 0$. According to \eqref{eq:hardysobholder} the Birman--Schwinger operator $V^{1/2}(D_\Phi^2)^{-1} V^{1/2}$ is well-defined and bounded. Changing to polar coordinates and logarithmic variables, this operator becomes $\widetilde V^{1/2}_+(-\partial^2_t+(\ii \partial_\theta-\Phi)^2)^{-1} \widetilde V^{1/2}_+$ in $L^2(\R\times\Sph^1)$, where
    $$
    \widetilde V(t,\theta)=e^{2t}V(e^t \cos\theta,e^t\sin\theta) \,.
    $$
    
    Applying the Birman--Schwinger principle (see, e.g., \cite[Subsection 4.3.3]{Frank2022SchrodingerInequalities}) and the Lieb--Thirring inequality (see \cite{Lieb1976InequalitiesInequalities} and also \cite[Theorem 4.59]{Frank2022SchrodingerInequalities}) we obtain that for $p=1+\alpha/2>1$ 
    \begin{equation}
        \begin{split}\label{eqn:thm1proofbirm}
        N(D_\Phi^2-V)&=n_+(1,\widetilde V^{1/2}(-\partial^2_t+(\ii \partial_\theta-\Phi)^2)^{-1} \widetilde V_+^{1/2})\\
        &\leq \Tr(\widetilde V^{1/2}(-\partial^2_t+(\ii \partial_\theta-\Phi)^2)^{-1} \widetilde V^{1/2})^p\\
        &\leq \Tr(\widetilde V^{p/2}(-\partial^2_t+(\ii \partial_\theta-\Phi)^2)^{-p} \widetilde V^{p/2}).
        \end{split}
    \end{equation}
    To compute the trace we need to find the integral kernel of the operator $(-\partial^2_t+(\ii \partial_\theta-\Phi)^2)^{-p}$, which we denote by $G_{\Phi,p}(t,\theta;\tau,\vartheta)$. We note that $(-\partial_t^2+(i\partial_\theta-\Phi)^2)$ in $L^2(\R\times\Sph)$ is unitarily equivalent, via a continuous and a discrete Fourier transform, to multiplication by $\xi^2+(n-\Phi)^2$ in $L^2(\R) \times \ell_2(\Z)$. Thus,
    $$
    G_{\Phi,p}(t,\theta;\tau,\vartheta) = \frac{1}{(2\pi)^2} \sum_{n\in\Z}\int_{\R}\frac{e^{\ii n(\theta -\vartheta)}e^{\ii\xi(t-\tau)}}{(\xi^2+(n-\Phi)^2)^p} \dd\xi \,.
    $$
    Given that $\Phi\in (0,1/2]$ and $p>1$ the above sum converges. Moreover, $g_{\Phi,p}:=G_{\Phi,p}(t,\theta;t,\theta)$ is independent of $t$ and $\theta$ and we compute that
    \begin{equation*}
        g_{\Phi,p}=\frac{\Gamma(p-1/2)}{4\pi^{3/2}\Gamma(p)}\sum_{n\in\Z} |n-\Phi|^{1-2p} = \frac{\Gamma((1+\alpha)/2)}{4\pi^{3/2}\Gamma(1+\alpha/2)}\sum_{n\in\Z} |n-\Phi|^{-1-\alpha} \,.
    \end{equation*}
    Returning to the estimate in \eqref{eqn:thm1proofbirm} we conclude that 
    \begin{align*}
        \Tr(\widetilde V_+^{p/2}(-\partial^2_t+(\ii \partial_\theta-\Phi)^2)^{-p} \widetilde V_+^{p/2})&= g_{\Phi,p} \int_{\R}\int_{-\pi}^\pi \widetilde V(t,\theta)^p_+\dd \theta\dd t\\
        &=g_{\Phi,p} \int_{\R^2} V(x)_+^{p}|x|^{2p-2}\dd x \,,
    \end{align*}
    which completes the proof of \eqref{eqn:ABweightedCLR} with the constant given in \eqref{eq:abexplicit}. This easily implies the upper bound in \eqref{eq:abequiv}. The lower bound is a consequence of the following remark. 
\end{proof}

\begin{remark} \label{rem:ABHardySobolev}
	A standard argument shows that the sharp constants in the CLR-type inequality \eqref{eqn:ABweightedCLR} and in the magnetic Hardy-Sobolev inequality \eqref{eqn:ABHardySobolev} satisfy
	\begin{equation}
		\label{eq:clrsob}
		S_{\Phi,q} \geq C_{\Phi,\alpha}^{-2/(\alpha+2)}
		\qquad\text{with}\ \frac{2}{\alpha+2} + \frac 2q = 1 \,.
	\end{equation}
	In particular, \eqref{eq:abexplicit} implies that
	$$
	S_{\Phi,q} \geq \left(\frac{\Gamma(1/2+2/(q-2))}{4\pi^{3/2}\Gamma(1+2/(q-2))}\sum_{n\in\Z}|n-\Phi|^{-1-4/(q-2)}\right)^{-(q-2)/q}
	$$
	and the upper bound in \eqref{eq:abequiv} implies that
	\begin{equation}
		\label{eq:abhardysoblower}
		S_{\Phi,q} \gtrsim_q d(\Phi)^{1+2/q} \,.
	\end{equation}
	Let us show that this bound is optimal, that is,
	\begin{equation}
		\label{eq:abhardysobupper}
		S_{\Phi,q} \lesssim_q d(\Phi)^{1+2/q} \,.
	\end{equation}
	In view of \eqref{eq:clrsob} this will prove the lower bound in \eqref{eq:abequiv} and thereby complete the proof of Theorem \ref{thm:ABweightedCLR}.
	
	We fix $\phi\in C^\infty_c(\R)$ and define
	$$
	u(r\cos\theta,r\sin\theta) = \phi((\ln r)/\ell)\, e^{in\theta} \,,
	$$
	where $n\in\Z$ is such that $d(\Phi) = |n-\Phi|$. Then, by \eqref{eqn:ABHardySobolev} after changing to logarithmic coordinates,
	\begin{align*}
		\ell^{-1} \int_\R |\phi'(t)|^2\dd t + d(\Phi)^2 \ell \int_\R |\phi(t)|^2 \dd t 
		\geq S_{\Phi,q} \left( \ell \int_\R |\phi(t)|^q\dd t \right)^{2/q}.
	\end{align*}
	Choosing $\ell = d(\Phi)^{-1}$ we obtain \eqref{eq:abhardysobupper}.
\end{remark}

\subsection{The antisymmetric operator}

The same construction and arguments carry over to the antisymmetric operator. In this case, the Hardy--Sobolev inequalities \eqref{eqn:ABHardySobolev} are replaced by the inequalities
\begin{equation}\label{eqn:AntiHardySobolev}
	\int_{\R^2}|\nabla u|^2\dd x \geq S_q \left(\int_{\R^2}\frac{|u|^q}{|x|^2}\dd x\right)^\frac{2}{q}
	\qquad\text{for all antisymmetric}\ u\in C^\infty_c(\R^2\setminus\{0\}) \,
\end{equation}
with $S_q>0$ provided that $q\in[2,\infty)$. A proof of this inequality can be found in \cite{Hoffmann-Ostenhof2021HardyFunctions}. In the special case $q=2$ we have
\begin{equation}\label{eqn:AntiHardy}
	\int_{\R^2}|\nabla u|^2\dd x \geq \int_{\R^2}\frac{|u|^2}{|x|^2}\dd x
	\qquad\text{for all antisymmetric}\ u\in C^\infty_c(\R^2\setminus\{0\})
\end{equation}
with the sharp constant equal to one.

For symmetric $V$ such that either the right side in Theorem \ref{thm:antiweightedCLR} or in Corollary \ref{cor:AntiweightedCLRweak} is finite we can define the operators $-\DeltaAnti-V$ in $L^2_{{\bf as}}(\R^2)$ similarly as in the Aharonov--Bohm case.

\begin{proof}[Proof of Theorem \ref{thm:antiweightedCLR}]
We fix $\alpha>0$ and take $0\leq V\in L^{1+\alpha/2}(\R^2;\dd x/|x|^2)$ as before. The Birman--Schwinger operator $V^{1/2}(-\DeltaAnti)^{-1}V^{1/2}$ in $L^2_{{\bf as}}(\R^2)$ is unitarily equivalent to the operator $\widetilde V^{1/2}(-\partial^2_t-\partial_\theta^2)^{-1} \widetilde V^{1/2}$ acting in the subspace of function $u\in L^2(\R\times\Sph^1)$ satisfying $u(t,\theta)=-u(t,\pi/2-\theta)$. Here $\widetilde V$ is defined as in the proof of Theorem \ref{thm:antiweightedCLR}. Applying the Birman--Schwinger principle and the Lieb--Thirring inequality as before, we are reduced to finding the integral kernel $G_p(t,\theta;\tau,\vartheta)$ corresponding to $(-\partial_t^2-\partial_{\theta}^2)^{-p}$ acting in this subspace. To find it, we argue as previously, using a Fourier decomposition in terms of the antisymmetric angular harmonics $\phi_n(\theta)=\pi^{-1/2}\sin(n(\theta-\pi/4))$, $n\in\N$. It follows that 
\begin{align*}
    G_p(t,\theta;t,\theta)&=\frac{1}{2\pi}\sum_{n=1}^\infty\phi_n(\theta)^2 \int_{\R}\frac{1}{(\xi^2+n^2)^p}\dd \xi\\
    &\leq \frac{\Gamma(p-1/2)}{2\pi^{3/2}\Gamma(p)}\left(\sum_{n=1}^\infty n^{1-2p}\right) = \frac{\Gamma(p-1/2)}{2\pi^{3/2}\Gamma(p)}\, \zeta(2p-1) \,,
\end{align*}
where $\zeta$ denotes the Riemann zeta function. This proves Theorem \ref{thm:antiweightedCLR}.
\end{proof}


\section{Interpolation and proof of corollaries \ref{cor:ABweightedCLRweak} and \ref{cor:AntiweightedCLRweak}} \label{sec:interpolation}

In this section we derive Corollaries \ref{cor:ABweightedCLRweak} and \ref{cor:AntiweightedCLRweak} from Theorems \ref{thm:ABweightedCLR} and \ref{thm:antiweightedCLR}, respectively. We use a variant of an interpolation argument by Birman and Solomyak \cite{Birman1988InterpolationOperator}, but we avoid any explicit mention of interpolation theory or ideals of compact operators.

\begin{proof}[Proof of Corollary \ref{cor:ABweightedCLRweak}]
	We fix $\alpha>0$ and recall that we may assume that $0<\Phi\leq 1/2$ and that $V\geq 0$. With two parameters $s>0$ and $0<\theta<1$ to be determined we write
	$$
	D_\Phi^2 - V = \theta( D_\Phi^2 - \theta^{-1} s |x|^{-2}) + (1-\theta) (D_\Phi^2 - (1-\theta)^{-1} |x|^{-2}(|x|^2 V-s)) \,.
	$$
	Assuming that $\theta^{-1} s\leq\Phi^2$ we can use the magnetic Hardy inequality \eqref{eqn:ABHardy} to bound
	$$
	D_\Phi^2 - V \geq (1-\theta) (D_\Phi^2 - (1-\theta)^{-1} |x|^{-2}(|x|^2 V-s)_+) \,.
	$$
	Thus, by the variational principle
	$$
	N(D_\Phi^2-V) \leq N(D_\Phi^2 - (1-\theta)^{-1} |x|^{-2}(|x|^2 V-s)_+) \,.
	$$
	For an arbitrary $0<\beta<\alpha$ we can apply Theorem \ref{thm:ABweightedCLR} and obtain
	$$
	N(D_\Phi^2-V) \leq C_{\Phi,\beta} (1-\theta)^{-1-\beta/2} \int_{\R^2} (|x|^2 V(x)-s)_+^{1+\beta/2}\,\frac{\dd x}{|x|^2} \,.
	$$
	Abbreviating $[V]:=\sup_{t>0} t^{1+\alpha/2} \int_{|x|^2V(x)>t}\frac{\dd x}{|x|^2}$ and using the layer cake representation we find
	\begin{align*}
		\int_{\R^2} (|x|^2 V(x)-s)_+^{1+\beta/2}\,\frac{\dd x}{|x|^2}
		& = (1+\beta/2) \int_0^\infty \int_{|x|^2 V(x)-s>\sigma}\, \frac{\dd x}{|x|^2} \,\sigma^{\beta/2}\,\dd\sigma \\
		& \leq (1+\beta/2)\, [V]\, \int_0^\infty (\sigma+s)^{-1-\alpha/2} \,\sigma^{\beta/2}\,\dd\sigma \\
		& = \frac{\Gamma(2+\beta/2)\,\Gamma((\alpha-\beta)/2)}{\Gamma(1+\alpha/2)}
		\, s^{(\beta-\alpha)/2} \,[V] \,. 
	\end{align*}
	In the last computation we used a beta function identity. To minimize this bound, we choose $s=\theta\Phi^2$ and obtain
	$$
	N(D_\Phi^2-V) \leq \frac{\Phi^{\beta-\alpha}\, C_{\Phi,\beta}}{\sup_{0<\theta<1}(1-\theta)^{1+\beta/2} \theta^{(\alpha-\beta)/2}}  \frac{\Gamma(2+\beta/2)\,\Gamma((\alpha-\beta)/2)}{\Gamma(1+\alpha/2)}
	\, [V] \,.
	$$
	This bound can still be optimized with respect to $\beta\in(0,\alpha)$. This proves \eqref{eqn:ABweightedCLRweak}. Taking a fixed $\beta$ (say $\beta=\alpha/2$) and recalling that $C_{\Phi,\beta}\lesssim_\beta \Phi^{-1-\beta}$ by \eqref{eq:abequiv}, we deduce the upper bound in \eqref{eq:abequivweak}. The lower bound follows from \eqref{eq:strongweak} together with the lower bound in~\eqref{eq:abequiv}.
\end{proof}

The proof of Corollary \ref{cor:AntiweightedCLRweak} is similar to that of Corollary \ref{cor:ABweightedCLRweak} and is omitted.

 
\section{Long-range potentials and behaviour of constants}\label{sec:asymptotics}

In this section we construct for arbitrary $\alpha>0$ a $V$, which in the strong coupling limit saturates the weak bounds \eqref{eqn:ABweightedCLRweak} and \eqref{eqn:AntiweightedCLRweak}. We follow arguments which were carried out for dimensions $d\geq 3$ in \cite{Birman1990NegativeDiscussion,Birman1992SchrodingerProblem,Laptev1993AsymptoticsConstant}.

\begin{theorem} \label{thm:ABWpasymp}
	Let $\Phi\in\R\setminus\Z$, let $p>0$ and assume that $V\in L^\infty(\R^2)$ satisfies
	$$
	V(x) =  |x|^{-2} (\ln|x|)^{-1/p} \left( 1+ o(1) \right)
	\qquad\text{as}\ |x|\to\infty \,.
	$$
	Then for $p>1$
    \begin{equation*}
        \lim_{\lambda\rightarrow\infty}\lambda^{-p}N(D_\Phi^2-\lambda V ) = \frac{\Gamma(p-1/2)}{2\sqrt{\pi}\Gamma(p)}\sum_{n\in\Z}\frac{1}{|n-\Phi|^{2p-1}} \,,
    \end{equation*}
    for $p=1$
    \begin{equation*}
        \lim_{\lambda\rightarrow\infty}(\lambda\ln\lambda)^{-1}N(D_\Phi^2-\lambda V)=\frac{1}{2} \,,
    \end{equation*}
	and for $p<1$
	\begin{equation*}
		\lim_{\lambda\rightarrow\infty}\lambda^{-1}N(D_\Phi^2-\lambda V)=
		\frac1{4\pi} \int_{\R^2} V(x)_+\dd x \,.
	\end{equation*}
\end{theorem}

In the theorem we clearly see the difference between the long range case $p\geq 1$ and the short range case $p<1$. In the former case the asymptotics are insensitive to the local behavior of $V$ and solely determined by its asymptotic behavior, while in the latter case they are essentially determined by the local behavior of $V$.

We note that if $V$ is as in the theorem with $p>1$, then with $\alpha =2(p-1)$
\begin{equation*}
	\lim_{\lambda\rightarrow\infty}\lambda^{-p}\sup_{t>0}t^{1+\alpha/2}\int_{\lambda |x|^2V(x)_+>t}\frac{\dd x}{|x|^2}=\sup_{t>0} t^{p}\int_{ |x|^2V(x)_+>t}\frac{\dd x}{|x|^2} \in (0,\infty) \,.
\end{equation*}
Therefore Theorem \ref{thm:ABWpasymp} shows that the weak bounds \eqref{eqn:ABweightedCLRweak} is saturated for the potentials $\lambda V$ as $\lambda\to\infty$.

Moreover, the asymptotics for $p=1$ show that one cannot expect to have a version of the weak inequality \eqref{eqn:ABweightedCLRweak} that is homogeneous of degree one in $V$. 

\begin{remark}
	For comparison, if $\Phi=0$ and $V$ is as in Theorem \ref{thm:ABWpasymp} with $p>1/2$ then $N(-\Delta-\lambda V)=\infty$ for all $\lambda>0$. The same holds for $p=1/2$ provided $\lambda>1/4$; see \cite[Proposition 4.21]{Frank2022SchrodingerInequalities}.
\end{remark}

\begin{proof}
	We mostly focus on the case $p\geq 1$ and discuss the case $p<1$ at the end. Let $W_p$ be defined as  
	\begin{align}\label{eqn:Wp}
		W_p(x) :=\begin{cases}
			|x|^{-2} (\ln|x|)^{-1/p}, &|x|>e \,,\\
			0,&|x|\leq e \,.
		\end{cases}
	\end{align}
	We will prove the theorem for $p\geq 1$ in the special case $V=W_p$. By simple approximation arguments, this implies the result in the general case.
	
We start by simplifying the problem. Consider the restriction of the operator $D_\Phi^2-\lambda W_p$ to the region $\{x\colon |x|>e\}$ with Dirichlet and Neumann boundary conditions, denoted by $H_{\Phi}^D(\lambda W_p)$ and $H_{\Phi}^N(\lambda W_p)$, respectively. Then, since $W_p\equiv 0$ for $|x|\leq e$, by the variational principle,
\begin{align}\label{eqn:bracketing}
    N(H_{\Phi}^D(\lambda W_p))\leq N(D_\Phi^2-\lambda W_p)\leq N(H_{\Phi}^N(\lambda W_p)) \,.
\end{align}
It follows, using logarithmic-coordinates $r=e^{t+1}$ and the definition of $W_p$, that we need only estimate the number of negative eigenvalues of the operator
\begin{equation*}
   -\partial_{t}^2+(\ii\partial_\theta-\Phi)^2-\lambda (t+1)^{-1/p} \text{ in } L^2((0,\infty)\times\Sph^1),
\end{equation*}
 from above and below, where the operator is considered with  Neumann and Dirichlet boundary conditions at $t=0$, respectively.

Now we carry out a further bracketing argument. We fix $L>0$ and for $k\in\N_0$ denote by $H_{k,L}^D(V)$ and $H_{k,L}^N(V)$ the restrictions of $-\partial_{t}^2+(\ii\partial_\theta-\Phi)^2-V(t)$ to the intervals $(kL,(k+1)L)$ with Dirichlet and Neumann boundary conditions respectively. Then, using $((k+1)L+1)^{-1/p}\leq (t+1)^{-1/p}\leq (kL+1)^{-1/p}$ on $(kL,(k+1)L)$,
\begin{align}
     N(H^D_{\Phi}(\lambda W_p))&\geq \sum_{k=0}^\infty N(H_{k,L}^D(\lambda (t+1)^{-1/p}))\geq   \sum_{k=0}^\infty N(H_{k,L}^D(\lambda ((k+1)L+1)^{-1/p})) \label{eqn:Wpboundbelow}
\end{align}
and
\begin{align}
     N(H^N_{\Phi}(\lambda W_p))&\leq \sum_{k=0}^\infty N(H_{k,L}^N(\lambda (t+1)^{-1/p}))\leq \sum_{k=0}^\infty N(H_{k,L}^N(\lambda (kL+1)^{-1/p})) \,. \label{eqn:Wpboundabove}
\end{align}
It remains to estimate each of these, where we first consider the case of $p>1$. Starting with the lower bound, we use \eqref{eqn:Wpboundbelow} to see that
\begin{align*}
    N(H^D_{\Phi}(\lambda W_p))&\geq\sum_{k=0}^\infty \#\{(m,n)\in\N\times\Z\colon \frac{\pi^2 m^2}{L^2}+(n-\Phi)^2<\lambda((k+1)L+1)^{-1/p})\}\\
    &\geq \sum_{m\in\N,n\in\Z}\left(L^{-1}\lambda^{p}\left(\pi^2 m^2/L^2+(n-\Phi)^2\right)^{-p}-1-L^{-1}\right)_{+}=\text{(I)}+\text{(II)} \,,
\end{align*}
where 
\begin{align*}
    \text{(I)}&=\sum_{m\in\N_0,n\in\Z}\left(L^{-1}\lambda^{p}\left(\pi^2 m^2/L^2+(n-\Phi)^2\right)^{-p}-1-L^{-1}\right)_{+}\\
    &\geq \lambda^p \sum_{n\in\Z}\int_{0}^\infty \left(\left(\pi^2 \tau^2+(n-\Phi)^2\right)^{-p}-\lambda^{-p}(L+1)\right)_{+}\dd \tau,
\end{align*}
and 
\begin{align*}
    \text{(II)}&=-\sum_{n\in\Z}\left(L^{-1}\lambda^{p}|n-\Phi|^{-2p}-1-L^{-1}\right)_{+}\\
    &\geq -L^{-1}\lambda^{p}\sum_{n\in\Z}|n-\Phi|^{-2p}.
\end{align*}
Meanwhile, for the upper-bound \eqref{eqn:Wpboundabove} we find that
\begin{align*}
    N(H^N_{\Phi}(\lambda W_p)) & \leq \sum_{k=0}^\infty \#\{(m,n)\in\N_{0}\times\Z\colon \frac{\pi^2 m^2}{L^2}+(n-\Phi)^2<\lambda(kL+1)^{-1/p} \} \\
    & =\text{(III)}+\text{(IV)},
\end{align*}
where 
\begin{align*}
    \text{(III)}&= \#\{(m,n)\in\N_{0}\times\Z\colon \frac{\pi^2 m^2}{L^2}+(n-\Phi)^2<\lambda \} \\
    & \leq \#\{ n\in \Z\colon (n-\Phi)^2<\lambda \}
    + \sum_{n\in\Z} \pi^{-1} L \left( \lambda - (n-\Phi)^2 \right)_+^{1/2} \\
    & \leq (2\sqrt\lambda +1) + 2\pi^{-1} L (\lambda-\Phi^2)_+^{1/2} + \pi^{-1} L \int_\R (\lambda - (t-\Phi)^2)_+^{1/2}\dd t \\ 
    & = (2\sqrt\lambda +1) + 2\pi^{-1} L (\lambda-\Phi^2)_+^{1/2} + 2^{-1} L \lambda
\end{align*}
and 
\begin{align*}
    \text{(IV)}&= \sum_{k=1}^\infty \#\{(m,n)\in\N_{0}\times\Z\colon \frac{\pi^2 m^2}{L^2}+(n-\Phi)^2<\lambda(kL+1)^{-1/p} \} \\
    & \leq \sum_{m\in\N_0, n\in\Z} \left( L^{-1}\lambda^p \left(\pi^2m^2/L^2+(n-\Phi)^2 \right)^{-p} - L^{-1} \right)_+ \\
    &\leq \lambda^p\sum_{n\in\Z}\int_0^\infty\left(\left(\pi^2 \tau^2+(n-\Phi)^2\right)^{-p}- \lambda^{-p} \right)_{+}\dd \tau.
\end{align*}
Taking the limsup and liminf as $\lambda\to\infty$ and then the limit $L\to\infty$, we find
\begin{align*}
    \liminf_{\lambda\rightarrow\infty}\lambda^{-p}N(H^D_{\Phi}(\lambda W_p))&\geq\sum_{n\in\Z}\int_0^\infty \left(\pi^2 \tau^2+(n-\Phi)^2\right)^{-p}\dd \tau\\
    &=\frac{\Gamma(p-1/2)}{2\sqrt{\pi}\Gamma(p)}\sum_{n\in\Z}\frac{1}{|n-\Phi|^{2p-1}},
\end{align*}
and similarly
\begin{align*}
    \limsup_{\lambda\rightarrow\infty}\lambda^{-p}N(H^N_{\Phi}(\lambda W_p))
    &\leq \frac{\Gamma(p-1/2)}{2\sqrt{\pi}\Gamma(p)}\sum_{n\in\Z}\frac{1}{|n-\Phi|^{2p-1}}.
\end{align*}
This proves the claimed bound for $p>1$.

For the case of $p=1$, we carefully consider the terms that produce a logarithmic divergence. In this case, the choice of intervals does not matter, so we take $L=1$. We start by using \eqref{eqn:Wpboundbelow} to find that
\begin{align*}
     N(H^D_{\Phi}(\lambda W_1))
    &\geq \lambda \sum_{m\in\N,n\in\Z}\left(\left(\pi^2 m^2+(n-\Phi)^2\right)^{-1}-2\lambda^{-1}\right)_{+}\\
    &\geq \lambda \int_{\R\backslash(-1,1)}\int_{1}^\infty \left(\left(\pi^2 \tau^2+(t-\Phi)^2\right)^{-1}-2\lambda^{-1}\right)_{+}\dd \tau\dd t-O(\lambda)\\
    &\geq \lambda (2\pi)^{-1}\iint_{\sigma^2+s^2>R_1^2} \left(\left(\sigma^2+s^2\right)^{-1}-2\lambda^{-1}\right)_{+}\dd \sigma\dd s-O(\lambda),
\end{align*}
with $R_1^2 := (\pi^2 + (1-\Phi)^2)/2$. When passing to the last line we increased the region of integration in the first term, noting that additional integral is $O(\lambda)$. For the upper bound \eqref{eqn:Wpboundabove}, in the decomposition above, the term $\text{(III)}$ is of order $O(\lambda)$ as $\lambda\rightarrow \infty$, thus we see that
\begin{align*}
    N(H^N_{\Phi}(\lambda W_1))&\leq O(\lambda)+\text{(IV)}\\
    &=\sum_{m\in\N_0,n\in\Z}\#\{k\in \N\colon k<\lambda \left(\pi^2 m^2+(n-\Phi)^2\right)^{-1}-1\}+O(\lambda)\\
    &= \lambda \sum_{m\in\N\backslash\{1\},n\in\Z}\left( \left(\pi^2 m^2+(n-\Phi)^2\right)^{-1}-\lambda^{-1}\right)_{+}+O(\lambda)\\
    &\leq \lambda (2\pi)^{-1} \iint_{\sigma^2+s^2>R_2^2} (( \sigma^2+s^2)^{-1}-\lambda^{-1})_{+}\dd\sigma\dd s+O(\lambda)
\end{align*}
with $R_2^2:=\pi^2 + (1-\Phi)^2$. For $R=R_1,R_2$ we compute
$$
\iint_{\sigma^2+s^2>R^2} (( \sigma^2+s^2)^{-1}-\lambda^{-1})_{+}\dd\sigma\dd s
= 2\pi \int_R^{\sqrt\lambda} (r^{-2} - \lambda^{-1}) r\dd r = \frac12\ln\lambda + O(1) \,.
$$
This proves the claimed bound for $p=1$.

Finally, we comment on the case $p<1$. We clearly have
\begin{align}
	\label{eq:weylupper}
	\liminf_{\lambda\to\infty} \lambda^{-1} N(D_\Phi^2-\lambda V) \geq \frac{1}{4\pi} \int_{\R^2} V(x)_+\,dx \,.
\end{align}
Indeed, for given, sufficiently large $R>0$ we bound $V\geq V\1(|x|<R)$ (here we use that $V$ is nonnegative outside of a bounded set) and then impose a Dirichlet condition at $|x|=R$ to bound $N(D_\Phi^2-\lambda V)$ from below by the number of negative eigenvalues of the corresponding Dirichlet operator on $\{|x|<R\}$. By \cite[Corollary 1.2]{FrLa} for the latter operator one has Weyl asymptotics. Since $R>0$ can be chosen arbitrarily large, we obtain \eqref{eq:weylupper}.

To prove
\begin{align}
	\label{eq:weyllower}
	\limsup_{\lambda\to\infty} \lambda^{-1} N(D_\Phi^2-\lambda V) \leq \frac{1}{4\pi} \int_{\R^2} V(x)_+\,dx \,,
\end{align}
we set, for all $R>1$, $\widetilde W_p(x) = \1(|x|>R) |x|^{-2} (\ln |x|)^{-1/p} + \1(|x|\leq R) R^{-2}(\ln R)^{-1/p}\}$. For $0<\theta<1$ and $\epsilon>0$ we decompose
$$
D_\Phi^2 -\lambda V = \theta \left(D_\Phi^2 - \theta^{-1}(1+\epsilon)\lambda \widetilde W_p \right) + (1-\theta) \left( D_\Phi^2 -(1-\theta)^{-1}\lambda(V-(1+\epsilon) \widetilde W_p) \right)
$$
and obtain
$$
N(D_\Phi^2 - \lambda V) \leq N(D_\Phi^2 - \theta^{-1}(1+\epsilon)\lambda \widetilde W_p) + N(D_\Phi^2 -(1-\theta)^{-1}\lambda(V-(1+\epsilon) \widetilde W_p)_+) \,.
$$
Since $\widetilde W_p$ is radially nonincreasing, it results from either \eqref{eqn:ABCLRBalinsky} or \eqref{eqn:ABCLRLaptev} that
$$
N(D_\Phi^2 - \theta^{-1}(1+\epsilon)\lambda \widetilde W_p)
\lesssim_\Phi \theta^{-1}(1+\epsilon)\lambda \int_{\R^2} \widetilde W_p\dd x
\lesssim_{\Phi,p}\theta^{-1}(1+\epsilon) (\ln R)^{1-1/p} \lambda \,.
$$
Meanwhile, by assumption there is an $R_\epsilon<\infty$ such that for all $|x|\geq R_\epsilon$ one has $V(x)\leq (1+\epsilon)|x|^{-2} (\ln |x|)^{-1/p}$. Therefore, the potential $(V-(1+\epsilon) \widetilde W_p)_+$ is supported in a ball and with the help of \cite{So} one finds
$$
\lim_{\lambda\to\infty} \lambda^{-1} N(D_\Phi^2 -(1-\theta)^{-1}\lambda(V-(1+\epsilon) \widetilde W_p)_+) = \frac{1}{4\pi} (1-\theta)^{-1} \int_{\R^2} (V-(1+\epsilon) \widetilde W_p)_+ \dd x \,.
$$
Thus, we have shown that
$$
\limsup_{\lambda\to\infty} \lambda^{-1} N(D_\Phi^2 - \lambda V) \leq \frac{1}{4\pi} (1-\theta)^{-1} \int_{\R^2} \! (V-(1+\epsilon) \widetilde W_p)_+ \dd x + C_{\Phi,p} \theta^{-1}(1+\epsilon) (\ln R)^{1-1/p}.
$$
Letting $R\to\infty$ using the integrability of $V$ and $p<1$, we obtain
$$
\limsup_{\lambda\to\infty} \lambda^{-1} N(D_\Phi^2 - \lambda V) \leq \frac{1}{4\pi} (1-\theta)^{-1} \int_{\R^2} V_+ \dd x \,.
$$
Since $\theta\in(0,1)$ is arbitrary, we obtain \eqref{eq:weyllower}. This concludes the proof.
\end{proof}

\begin{remark}
	Let us use Theorem \ref{thm:ABWpasymp} to prove the lower bound \eqref{eq:abequivweak} on $C_{\Phi,\alpha}'$. Let $\alpha>0$ and $p=1+\alpha/2>1$, then for $W_p$ as in the proof of Theorem \ref{thm:ABWpasymp}
	\begin{equation*}
		\lim_{\lambda\rightarrow \infty} \lambda^{-p} \sup_{t>0}t^{p}\int_{\lambda W_p|x|^2>t}\frac{\dd x}{|x|^2}=\sup_{t>0}t^{p}\int_{ W_p|x|^2>t}\frac{\dd x}{|x|^2}=2\pi,
	\end{equation*}
	and thus, by the asymptotic formula in Theorem \ref{thm:ABWpasymp},
	\begin{equation*}
		C_{\Phi,\alpha}^\prime\geq \lim_{\lambda\rightarrow \infty}\frac{N(D_\Phi^2-\lambda W_p)}{\sup_{t>0}t^{p}\int_{\lambda W_p|x|^2>t}\frac{\dd x}{|x|^2}}=\frac{\Gamma(\alpha/2+1/2)}{4\pi^{3/2}\Gamma(1+\alpha/2)}\sum_{n\in\Z}|n-\Phi|^{-1-\alpha}.
	\end{equation*}
	This proves \eqref{eq:abequivweak}.
\end{remark}

Finally, we note that the corresponding results hold in the antisymmetric case by near identical argument. We state them below without proof.

\begin{theorem}
	\label{thm:antiWpasymp}
	Let $p\geq 1$ and let $V$ be as in Theorem \ref{thm:ABWpasymp}.	Then for $p>1$
	\begin{equation*}
		\lim_{\lambda\rightarrow\infty}\lambda^{-p}N(-\DeltaAnti-\lambda V ) = \frac{\Gamma(p-1/2)}{\sqrt{\pi}\Gamma(p)} \,\zeta(2p-1)
	\end{equation*}
	and for $p=1$
	\begin{equation*}
		\lim_{\lambda\rightarrow\infty}(\lambda\ln\lambda)^{-1}N(-\DeltaAnti-\lambda V)=\frac{1}{2} \,.
	\end{equation*}
\end{theorem}

\subsection*{Acknowledgments}
Partial support through US National Science Foundation grant DMS-1954995 (R.L.F), as well as through the Excellence Strategy of the German Research Foundation grant EXC-2111-390814868 (R.L.F.) is acknowledged.


\end{document}